\documentclass[12pt]{article}
\usepackage[breaklinks=true]{hyperref}
\usepackage{graphicx}
\usepackage{epsfig}
\usepackage{dcolumn}
\usepackage{subfigure}
\usepackage{amsmath,amsthm, amssymb}
\usepackage{slashed}

\usepackage{comment}

\newcommand{\bea}{\begin{eqnarray}}
\newcommand{\eea}{\end{eqnarray}}
\newcommand{\beq}{\begin{equation}}
\newcommand{\eeq}{\end{equation}}
\newcommand{\nn}{\nonumber}

\newcommand{\half}{{\frac{1}{2}}}

\newcommand{\ket}[1]{\,\vert #1\rangle}

\newcommand{\avg}[1]{\left\langle #1 \right\rangle}

\newcommand{\Tr}{\mathrm{Tr}}

\def\be{\begin{equation}}
\def\ee{\end{equation}}
\def\beq{\begin{eqnarray}}
\def\eeq{\end{eqnarray}}

\newcommand{\C}[1]{\mathcal{#1}}

\newtheorem{theorem}{Theorem}
\newtheorem{lemma}[theorem]{Lemma}

\begin{document}

\begin{center}
\vspace{48pt}
{ \Large \bf FZZT Brane Relations in the Presence of Boundary Magnetic Fields}

\vspace{40pt}
{\sl Max R. Atkin}$\,^{a}$,
and {\sl Stefan Zohren}$\,^{b}$  
\vspace{24pt}

{\small

$^a$~Fakult\"{a}t f\"{u}r Physik, Universit\"{a}t Bielefeld,\\
Postfach 100131, D-33501 Bielefeld, Germany

\vspace{10pt}

$^b$~Department of Physics, Pontifica Universidade Cat\'olica do Rio de Janeiro,\\
Rua Marqu\^es de S\~ao Vicente 225, 22451-900 G\'avea, Rio de Janeiro, Brazil\\
and\\
Rudolf Peierls Centre for Theoretical Physics,\\
1 Keble Road, Oxford OX1 3NP, UK.
}

\end{center}


\vspace{36pt}

\begin{center}
{\bf Abstract}
\end{center}
We show how a boundary state different from the $(1,1)$ Cardy state may be realised in the $(p,p+1)$ minimal string by the introduction of an auxiliary matrix into the standard two hermitian matrix model. This boundary is a natural generalisation of the free spin boundary state in the Ising model. The resolvent for the auxiliary matrix is computed using an extension of the saddle-point method of Zinn-Justin to the case of non-identical potentials. The structure of the saddle-point equations result in a Seiberg-Shih like relation between the boundary states which is valid away from the continuum limit, in addition to an expression for the spectral curve of the free spin boundary state. We then show how the technique may be used to analyse boundary states corresponding to a boundary magnetic field, thereby allowing us to generalise the work of Carroll et al. on the boundary renormalisation flow of the Ising model, to any $(p,p+1)$ model. 
\vspace{12pt}
\vfill

{\footnotesize
\noindent
$^a$~{email: matkin@physik.uni-bielefeld.de}\\
$^b$~{email: zohren@fis.puc-rio.br}\\

}


\newpage

\section{Introduction}

Defining quantum gravity by a path integral over geometries has been a technically difficult idea to realise. One approach to making such an idea precise is known as Dynamical Triangulation (DT), which consists of approximating each geometry by a discrete geometry built from gluing together elementary building blocks, most often triangles, or higher dimensional analogues. The program of DT is most often carried out in two dimension, however it was hoped for a number of years that DT would provide a definition for quantum gravity in dimensions greater than two. This unfortunately did not turn out to be the case (for an overview consult \cite{Ambjorn:1997di,Ginsparg:1993is}) and attention has since shifted to the use of a refined method known as CDT, introduced by Ambj{\o}rn and Loll \cite{Ambjorn:1998xu}.

The theory of DT in two dimensions is nevertheless interesting, as the sum over two dimensional random surfaces may be reinterpreted as the worldsheet theory for the Polyakov string \cite{original}. The choice of matter coupled to gravity is then interpreted as describing the background in which the string is embedded. Furthermore, in two dimensions there exist analytical methods for studying the theory. Firstly, the continuum path integral over geometry can be shown to be equivalent to quantum Liouville theory, which is integrable. Secondly, the sum over the discrete planar graphs required in the DT approach may be explicitly performed via matrix integrals, thereby allowing an analytic investigation of the properties of DT (for a review of these topics consult \cite{Ginsparg:1993is, Nakayama:2004vk}). Thus we have two complementary approaches in two dimensions. An important step in relating these two approaches was the work by Kazakov and Boulatov which provided a two-matrix model formulation of the Ising model coupled to DT \cite{original5,original6} (see also \cite{original7,original8}) and the seminal works of Knizhnik, Polyakov and Zamolodchikov (KPZ) \cite{original2} and David \cite{original3} and Distler and Kawai \cite{original4}. Additional results, in particular, the formulation of the double scaling limit, came from the works of Br\'ezin and Kazakov \cite{original9}, Douglas and Shenker \cite{original10} and Gross and Migdal \cite{original11}.


In general, adding a statistical mechanical system to the DT triangulations or equivalently the addition of a minimal model $(p,q)$, as introduced by Belavin, Polyakov and Zamolodchikov \cite{Belavin:1984vu}, to the Liouville theory leads to what is know as a $(p,q)$ minimal string \cite{original12,original13,original14}. This point of view was responsible for much of the early interest in DT, and led to the discovery of many connections to integrable hierarchies \cite{Ginsparg:1993is}. In this perspective boundaries play an important role as they correspond to branes. In the continuum formulation of Liouville theory coupled to a CFT, two forms of boundary states for the gravitational degrees of freedom have been identified; Fateev-Zamolodchikov-Zamolodchikov-Teschner (FZZT) \cite{Fateev:2000ik, Ponsot:2001ng} and Zamolodchikov-Zamolodchikov (ZZ) \cite{Zamolodchikov:2001ah} which depend on a real parameter $\sigma$ and integer valued parameters $n$ and $m$ respectively. When the FZZT boundary condition is tensored with a Cardy state for the minimal model degrees of freedom, one obtains a brane known as a FZZT $(k,l)$ brane, denoted by $\ket{\sigma;k,l}$ .
On the DT side however this situation is not as clear. A two hermitian matrix model of the form,
\beq
Z = \int [dX][dY] \exp(-\frac{N}{g} \Tr[V_1(X) + V_2(Y) - XY]),
\eeq
is required in order for the scaling limit to realise a $(p,q)$ CFT interacting with Liouville theory. However, this matrix model has only two observables that correspond to boundary states; the resolvents $W_X(x) \equiv \avg{\frac{1}{N} \Tr \frac{1}{x-X} }$ and $W_Y(y) \equiv \avg{\frac{1}{N} \Tr \frac{1}{y-Y} }$. These resolvents act as generating functions for boundaries of finite size. The existence of only two obvious boundary states is in contrast to the $(p-1)(q-1)/2$ boundary states one can obtain by tensoring a Cardy state with the FZZT boundary condition. 

The main suggestion for resolving this issue was given in \cite{Seiberg:2003nm}, in which it was noted that if the boundary state is inserted into a spherical worldsheet (thereby creating a disc), then the following identification of states holds,
\beq \label{SSequiv}
\ket{\sigma;k,l}=\sum^{k-1}_{n=-(k-1),2} \sum^{l-1}_{m=-(l-1),2}  \ket{\sigma+i\frac{mp+nq}{\sqrt{pq}};1,1},
\eeq 
which is often referred to as the Seiberg-Shih relation. Since it was noted that the resolvent for $X$ produced the disc amplitude corresponding to insertion of the state $\ket{\sigma;1,1}$ it was therefore claimed there was no contradiction between the DT and Liouville approaches (see also \cite{Hosomichi:2008th}). This was later challenged by a number of authors \cite{Gesser:2010fi,Atkin:2010yv}, who observed that such an identification of states no longer held when the state was inserted into a worldsheet of more complicated topology. However, the issue of whether this is indeed the case is yet to be entirely settled \cite{Oh:2011ig}. Despite this, we feel it is still interesting to understand if the other boundary states may be computed directly from the matrix model and secondly to understand where the Seiberg-Shih relation appears in the matrix model in the planar limit. These questions have been partially addressed by \cite{Ishiki:2010wb} 
in which they implicitly argue the Seiberg-Shih relations should hold away form the scaling limit. However, they do not address the construction of the spectral curve nor how to relate their construction to the topological recursion relations of Eynard and Orantin \cite{eynardreview}. Finally, the physical interpretation of the boundary states constructed therein is less clear.

In this paper we provide a different perspective on the solution to these questions, which generalises the work of \cite{Carroll:1996xi,Atkin:2010yv} in which a boundary magnetic field applied to an Ising model on a fluctuating lattice was studied. We show how one can use saddle point equations to study a boundary magnetic field applied to any $(p,p+1)$ minimal string. In performing this study we will see Seiberg-Shih like equations arise naturally at the discrete level before a scaling limit is taken. Such relations naturally generate algebraic curves to which the construction of \cite{eynardreview} could be applied. For a non-zero boundary magnetic field we are able to extend the results of \cite{Carroll:1996xi,Atkin:2010yv}, in which fixed points of the boundary renormalisation flow were found, corresponding to fixed and free spin boundary states, to all $(p,p+1)$ minimal strings. 

\section{Free spin boundary for the $(p,p+1)$ minimal string}
Recall the two hermitian matrix model,
\beq
\label{2MM}
Z = \int [dX][dY] \exp(-\frac{N}{g} \Tr[V_1(X) + V_2(Y) - XY]),
\eeq
where the potential $V_1$ is of order $p$ and $V_2$ is of order $q$. The $(p,q)$ minimal model coupled to Liouville gravity is the continuum theory corresponding to the scaling limit of this two hermitian matrix model about its highest order multi-critical point. Furthermore the $(1,1)$ FZZT brane may be realised in the matrix model by the resolvent for $X$. The resolvent for $Y$ corresponds to the dual FZZT brane \cite{Seiberg:2003nm}. These resolvents may be computed using the master loop equation or spectral curve \cite{Eynard:2002kg},
\beq
\label{mastereqn}
E(x,y) = \left(V_1'(x) - y\right)\left(V_2'(y) - x \right) + g P(x,y) = 0,
\eeq
where the resolvent $W_X(x) \equiv \frac{1}{g}(U'(x) - y(x))$ and $P(x,y)$ is a polynomial in $x$ and $y$ with undetermined coefficients that require extra analytic data (such as the one-cut assumption) to be determined.

The $(p,p+1)$ minimal models coupled to Liouville gravity may be obtained in a second way, distinct from that given above, by having $V_1$ and $V_2$ identical and both of order $p$. This second way of realising the $(p,p+1)$ models will be of central importance in this paper. In particular for the case $p=3$, which corresponds to the Ising model, there is a clear interpretation of the matrix degrees of freedom as Ising spins; $X$ and $Y$ vertices in a Feynman diagram correspond to up and down spins respectively. With this interpretation it is easy to construct a representation of the Cardy state corresponding to free spin \cite{Carroll:1995nj}; it corresponds to the $X+Y$ resolvent, $W_{X+Y}(z) \equiv \avg{\frac{1}{N} \Tr \frac{1}{z-(X+Y)} }$. 

For $p>3$ we expect the $X+Y$ resolvent to again flow to a conformal boundary condition in the scaling limit that differs from the identity boundary condition obtained from the $X$ resolvent. We will refer to this boundary condition as the {\emph{free spin boundary state}}.

The problem of calculating such a resolvent for the case $p=3$ can be solved by simply changing variables in \eqref{2MM} to $S=X+Y$ and $A=X-Y$. This converts \eqref{2MM} to a $O(1)$ model and the $X+Y$ resolvent into a standard resolvent of the matrix $S$, which may then be computed for example via loop equations \cite{Carroll:1995nj,Kostov:1988fy}. However, this method does not generalise in an easy way to larger $p$. Instead we will now present an alternative calculation of the free spin boundary condition of the Ising model which is ripe for generalisation.

\begin{figure}[t]
\centering 
\parbox{7cm}{
\includegraphics[scale=0.45]{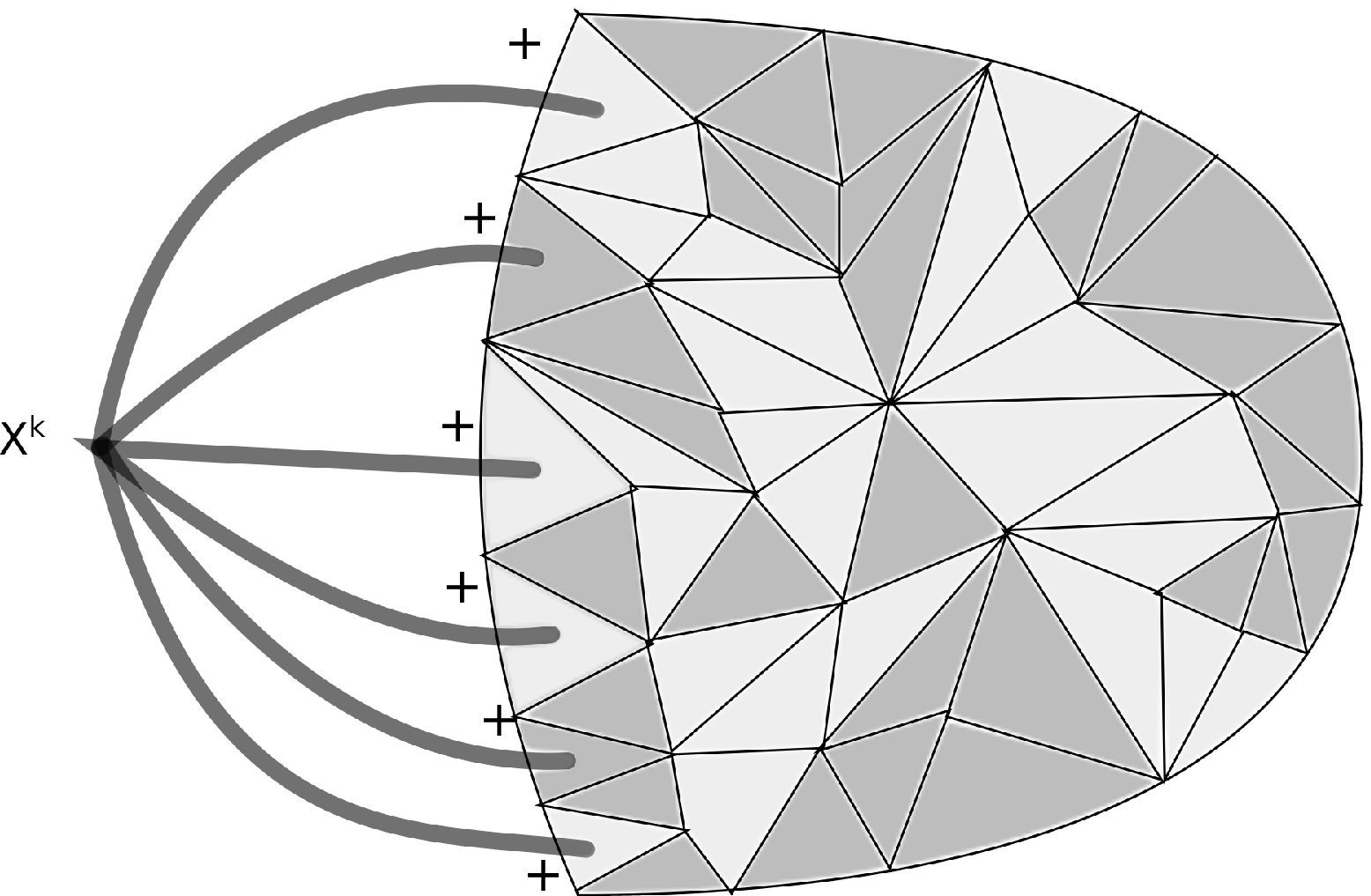}
\caption{An example of a Feynman diagram contributing to $\avg{\Tr X^k}$. In the bulk we have drawn the dual graph, with faces coloured dark and light grey corresponding to $X$ and $Y$ vertices respectively. The external $X^k$ vertex enforces an up spin to reside on the external side of each boundary link.}
\label{feynbound0}}
\qquad 
\begin{minipage}{7cm}
\includegraphics[scale=0.45]{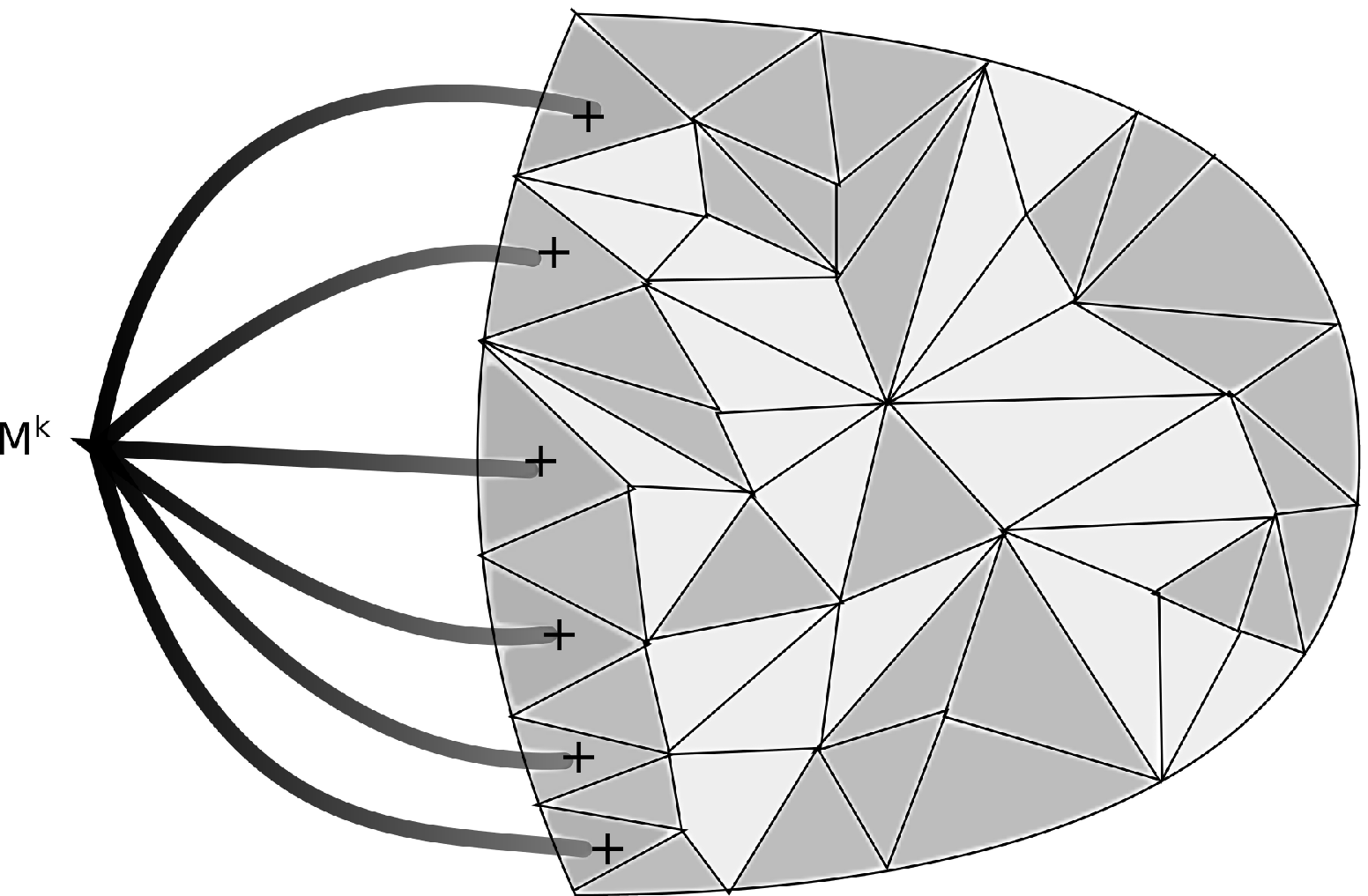}
\caption{An example of a Feynman diagram contributing to $\avg{\Tr M^k}$ in the case $M$ only couples to $X$. The $M \rightarrow X$ propagator is denoted by a black-to-dark grey line. We see that the external $M^k$ vertex enforces an up spin to reside on each of the faces touching the boundary.}
\label{feynbound1} 
\end{minipage} 
\end{figure}

First we consider $W_X$, this enforces a fixed spin boundary condition by ensuring that a spin residing in a face touching the boundary sees an up spin on the other side of the boundary edge as shown in Figure \ref{feynbound0}. It does not however put any restriction on the spins appearing within the triangulation. In contrast, we now proceed by making the assumption that a boundary condition can equally well be introduced by imposing a constraint on the spins residing in faces touching the boundary as shown in Figure \ref{feynbound1}. This form of constraint can be imposed by introducing another matrix into the matrix model action which when integrated out yields the original matrix model. For example consider the following model,\footnote{In this section we always denote an arbitrary potential by $V_i$ without necessarily meaning they are the same potentials in each expression.}
\beq
\label{MMwithM}
Z = \int [dM][dX][dY] \exp \left(-\frac{N}{g}\Tr[V_1(X) + V_2(Y) + \frac{\lambda}{2} M^2 - XY - XM ] \right),
\eeq
Note that $M$ only has the bulk effect of renormalising the $X$ propagator. However if we compute the resolvent of $M$ it will produce Feynman diagrams, as show in Figure \ref{feynbound1}, in which faces lying on the boundary may only contain up spins. In this example we expect the $M$ resolvent will be equal to the $X$ resolvent, up to non-universal terms.

Generalising this idea, we see we can introduce an extra matrix $M$, which by coupling equally to $X$ and $Y$, will have a resolvent computed by summing over all configurations of spins in the faces bordering the boundary as shown in Figure \ref{feynbound2}. Explicitly we have,
\beq
\label{genBS}
Z = \int [dM][dX][dY] \exp \left(-N\Tr[V_1(X) + V_2(Y) + \frac{\lambda}{2} M^2 - XY - (X + Y)M ] \right),
\eeq
where we have absorbed the $g$ dependence in to $\lambda$. We expect that resolvent for $M$ will flow to the same boundary condition as the $X+Y$ resolvent, in the scaling limit.

\begin{figure}[t]
\centering 
\parbox{15cm}{
\centering 
\includegraphics[scale=0.5]{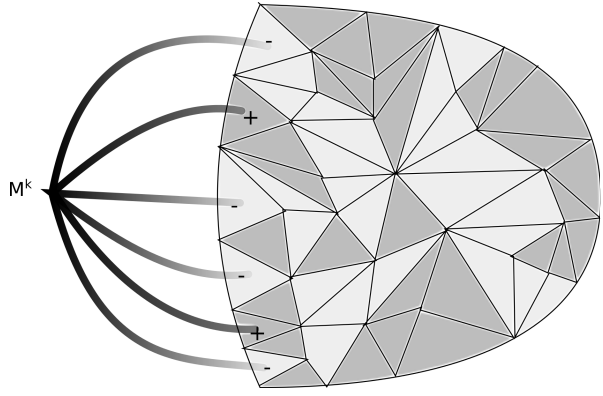}
\caption{By allowing $M$ to couple to $X$ and $Y$ equally, we have that $\avg{ \Tr M^k}$ will be computed by summing over diagrams in which up and down spins appear with equal weight on the boundary. Here the black-to-light-grey line represent $M \rightarrow Y$ propagators.}
\label{feynbound2}}
\end{figure}

The matrix model in \eqref{genBS} corresponds to a 3-state Potts model with non-equal potentials. The Potts model has been solved via loop equations but only for all potentials identical \cite{Eynard:1999gp}. However, there exists an alternate solution of the Potts model based on saddle-point methods which can be generalised to allow for non-identical potentials. Firstly, note that the above model may be diagonalised by the trick of Daul \cite{Daul:1994qy}; introduce a new matrix $A$ such that the model may be written as,
\beq
\label{genBSA}
Z = \!\int [dM] [dA] [dX][dY] \exp\! \left(-N \Tr[V_1(X) + V_2(Y) + \half M^2\! - \! (X + Y + M)A + \frac{\gamma}{2} A^2] \right)\! ,
\eeq
where $V_1(x) = \sum_{k = 1}^p \frac{t_k}{k} x^k + \frac{1}{2g}x^2$ and $V_2(x) = \sum_{k = 1}^q \frac{t'_k}{k} x^k+ \frac{1}{2g}x^2$. Note that $g$ dependence is now absorbed in $\gamma$. We now may integrate out the angular degrees of freedom using the Itzykson-Zuber integral applied to each of the $XA$, $YA$ and $MA$ terms. The resulting model falls into a general class of models of the form,
\beq
\label{genQmodel}
Z = \int [dA]\prod^Q_i [dB^{(i)}] \exp \left(-N \Tr[\sum^Q_i V_i(B^{(i)}) + V_A(A) - A\sum_i B^{(i)}] \right).
\eeq

In a work by Zinn-Justin \cite{ZinnJustin:1999jg} the above model was solved for $Q<5$ in the case $V_i = V$ for all $i$, via a saddle point method. For our purposes we need to solve this class of models with unequal potential, fortunately the generalisation to the case of differing $V_i$ is simple. First we review the method of Zinn-Justin \cite{ZinnJustin:1999jg} for the case $V_i = V$ for all $i$, for which the integral may be written as,
\beq
\label{identpot}
Z = \int [dA] e^{-N \Tr V_A(A)}\left(\Xi[A]  \right)^Q.
\eeq
where
\beq
\Xi[A] = \int [dB] e^{N \Tr \left[-V(B) + AB \right]}.
\eeq
If we compute the saddle point equations for the eigenvalues $a_i$ of $A$, we will encounter the function $\frac{1}{N}\frac{\partial}{\partial a_j} \log \Xi[A]$. Following \cite{ZinnJustin:1999jg} we introduce the resolvent for $A$, $W_A(a)$ and a new function $b(a)$ defined such that for large $N$ we may write,
\beq
\label{bcuteqn}
\frac{1}{N}\frac{\partial}{\partial a_j} \log \Xi[A] = \slashed{b}(a_j) - \slashed{W}_A(a_j),
\eeq
where $\slashed{f}(z) \equiv \half(f(z+i0) + f(z-i0))$. Note that the potential $V(B)$ plays no role in this saddle point equation, the resolvent $W_A$ arises from the variation of the Vandemont determinant and $b(a)$ is due to the variation of the $AB$ interaction term in the action. Varying $\Xi[A]$ now with respect to the eigenvalues of $B$, $b_i$, one obtains,
\beq
\label{acuteqn}
\frac{1}{N}\frac{\partial}{\partial b_j} \log \Xi[A] = \slashed{a}(b_j) + \slashed{W}_B(b_j) - V'(b_j),
\eeq
where we have introduced a new function $a(b)$ to account for the variation of the $AB$ term with respect to the eigenvalues of $B$. The crucial result of \cite{ZinnJustin:1999jg} is that $b(a)$ possess the following properties,
\begin{itemize}
\item The function $b(a)$ has the same cut as the physical cut of the $W_A$ resolvent, i.e. the location of the cut is the same and the difference across the cut is identical for the two functions.
\item The function $b(a)$ is the functional inverse of $a(b)$.
\end{itemize}

We now introduce some notation. The sheet of $b(a)$ with the same cut as $W_A$ we denote $b_0(a)$ and the sheet joined to this via the physical cut we call $b_{0\ast}(a)$. It will turn out that the sheets $b_{0}(a)$ and $b_{0\ast}(a)$ themselves are joined to other sheets via branch cuts ending at infinity as shown in Figure \ref{sheets0}. We denote sheets reached by traversing these cuts by $b_{i}(a)$ and $b_{i\ast}(a)$ respectively. This structure of the Riemann surface will actually be common to many of the functions considered in this paper and we shall use the notation for the sheets of such functions.

We now proceed by demonstrating that the analytic structure of $b(a)$ is indeed as outlined above and compute the asymptotic behaviour of $b(a)$ on each sheet using the saddle point equations for $b$. With this knowledge we can compute both the $A$ and $B^{(i)}$ resolvents.


\subsection{The function $b(a)$ below the physical cut}
\label{belowcut}
The saddle point equation of $\Xi[A]$ for $b$ may be written as,
\beq
\half(W_B(b)+W_{B\ast}(b)) + \half(a_0(b)+a_{0\ast}(b)) = V'(b),
\eeq
which using the fact that $a(b)$ shares the same physical cut with $W_B$, (see \cite{ZJ0}), we can rewrite as,
\beq
\label{bbelowcut}
W_B(b) + a_{0\ast}(b) = V'(b).
\eeq
Since $a$ and $b$ are inverse functions, this may be written as,
\beq
\label{bbelowcut2}
W_B(b(a)) + a = V'(b(a)),
\eeq
which using the known asymptotic form as $b \rightarrow \infty$ of $W_B(b)$, $W_B(b) \sim 1/b$, allows one to explicitly compute the asymptotic form of $b(a)$ as $a$ goes to infinity on each sheet below the physical cut. In particular if $V$ is of order $p$, we see that $b$ has a $p-1$ order branch point at infinity. We therefore know ``half'' the analytic structure of $b(a)$. We say ``half'' the analytic structure as we know $b(a)$ has a branch cut corresponding to the physical cut of $W_A(a)$, however we have yet to uncover the form of the Riemann surface behind this cut. We therefore have the situation illustrated in Figure \ref{sheets0}, in which the other side of the physical cut is still unknown.

\begin{figure}[t]
\centering 
\parbox{15cm}{
\centering 
\includegraphics[scale=0.5]{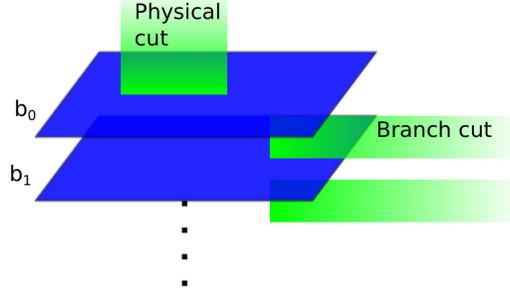}
\caption{The structure of the Reimann surface of $b(a)$ below the physical cut. The structure of the Riemann surface above the physical cut is left undetermined by the saddle point equations for $b$.}
\label{sheets0}}
\end{figure}

\subsection{The function $b(a)$ above the physical cut}
To determine the analytic structure of the function on the other side of the physical cut we must utilise the saddle point equation for $a$, which for general $Q$ is,
\beq
(W_A(a)+W_{A\ast}(a)) + Q(b(a)- W_A(a)) = V_A'(a)
\eeq
which may be written,
\beq
\label{Qsaddlept}
(2-Q)W_A(a) + b_{0\ast}(a) + (Q-1) b(a) = V_A'(a).
\eeq
Using \eqref{Qsaddlept} allows one to obtain the asymptotic behaviour of $b(a)$ as $a \rightarrow \infty$, on any sheet above the physical cut in terms of the asymptotic behaviour of $b(a)$ on the sheets below the physical cut, where by up and down we mean the vertical direction in Figure \ref{sheets0}. Given the asymptotic behaviour of $b$ in each sheet as $a \rightarrow \infty$, an algebraic equation for $b$ can be written down in which each coefficient is a symmetric function of the value of $b$ on each sheet. More explicitly, suppose that $b$ has $n$ sheets, we make the ansatz that it arises as the solution of an algebraic equation of order $n$. Furthermore, it must have the form,
\beq
\label{bloop}
b^n + \left( \sum^p_{i = 0} b_i + \sum^{n-p-2}_{i = 0} b_{i\ast} \right) b^{n-1} + \ldots +  \prod^p_{i = 0} b_i \prod^{n-p-2}_{i = 0} b_{i\ast} = 0,
\eeq
where we recall that $b_i$ is the function $b(a)$ restricted to the $i$th sheet below the physical cut with $b_0$ being the sheet containing the physical cut and $b_{i\ast}$ is the $i$th sheet above the physical cut with $b_{0\ast}$ containing the physical cut. By introducing $p$ we have allowed for the possibility of there being a different number of sheets below and above the physical cut. The coefficients of \eqref{bloop} are symmetric combinations of the values of $b$ on the different sheets. If this is to be an algebraic equation then these symmetric combinations should be polynomials in $a$. Since we know the explicit asymptotic expansion of $b_i$ and $b_{i\ast}$ we can compute the asymptotic expansion of each of the coefficients in \eqref{bloop}. One can then see explicitly that either their asymptotic expansion terminates, giving only polynomial terms in $a$, or undetermined coefficients appear at positive powers of $a$, related to the coefficients in the large $b$ expansion of $W_B(b)$. The requirement that negative powers of $a$ vanish gives relations among these coefficients.

This algebraic curve corresponds to the curve found via loop equation methods \cite{Eynard:2002kg} and is known as the spectral curve. The spectral curve forms part of the initial information required to begin the topological recursion \cite{eynardreview} which allows for the computation of all large $N$ corrections.

To generalise \eqref{Qsaddlept} to allow the potentials to differ, as is required to solve \eqref{genQmodel}, for which \eqref{genBS} is a special case, we simply redo the analysis of the saddle point equation leading to \eqref{bbelowcut}, for each matrix $B^{(i)}$. Associated to each matrix $B^{(i)}$ will be a function $b^{(i)}(a)$ and the potential for the matrix $B^{(i)}$ will determine the analytic structure and asymptotic behaviour for $b^{(i)}(a)$ on each sheet below the physical cut. Once this information is known for each $b^{(i)}(a)$, we then use the saddle-point equation for $a$ to determine the full analytic structure and asymptotic behaviour of each $b^{(i)}(a)$ on all sheets.

To motivate the form for the saddle point equations for $a$ it is instructive to briefly consider the case of $Q=2$.
In this case we have three matrices, $A$, $B^{(1)}$ and $B^{(2)}$. Having obtained the behaviour of the functions $b^{(1)}(a)$ and $b^{(2)}(a)$ below the physical cut using the saddle point equation for $b^{(1)}$ and $b^{(2)}$, we may find $b^{(1)}(a)$ and $b^{(2)}(a)$ above the physical cut using the saddle point equation for $a$ which generalises \eqref{Qsaddlept},
\beq
(W_A(a)+W_{A\ast}(a)) + (b^{(1)}_0(a) - W_A(a)) + (b^{(2)}_0(a) - W_A(a)) = V_A'(a),
\eeq
which can be written in two ways;
\beq
\label{bsaddleeqn}
{b}_{0\ast}^{(1)}(a) + b^{(2)}_0(a) = V_A'(a)
\eeq
or
\beq
{b}_{0\ast}^{(2)}(a) + b^{(1)}_0(a) = V_A'(a).
\eeq
In the general case when we have matrices $A$ and $B^{(i)}$ where $i \in \{1, \ldots, Q\}$ all with different potentials one has a set of equations indexed by $j$,
\beq
\label{Qsaddle}
(2 - Q)W_A(a) + b^{(j)}_{0\ast}(a) + \sum_{i\neq j}b^{(i)}_0(a) = V_A'(a),
\eeq
from which one can obtain the value of $b^{(i)}$ on any sheet and therefore construct its algebraic curve.

In the following, when working with a fixed $Q$ and non-identical potentials, we will often want to avoid the bulky notation $B^{(i)}$ for each of the matrix degrees of freedom in \eqref{genQmodel}. In such cases the function $b^{(i)}(a)$ will be denoted by a lower case of the corresponding matrix in the action. For example if $Q=2$ and we relabel $B^{(1)} \rightarrow X$ and $B^{(2)} \rightarrow Y$ then the associated $b^{(i)}(a)$ functions will be denoted $x(a)$ and $y(a)$ respectively.

\subsection{A Digression}
At this point it is useful to follow a slight digression which will produce some results useful in later sections of this paper. We have argued that the $M$ resolvent for the matrix model \eqref{genBS} will correspond to a free spin boundary state, however naively one might imagine we could make do with a simpler construction in which the free spin boundary is given by the $A$ resolvent of the model,\footnote{Note we have chosen the coefficient of $A^2$ such that after integrating out $A$ we return to the standard form of the action.}
\beq
\label{MMwithA}
Z = \int [dA][dX][dY] \exp \left(-N \Tr[V(X) + V(Y) + \frac{g}{2} A^2 - (X + Y)A ] \right).
\eeq
This expectation is incorrect since in the above model any interaction between $X$ and $Y$ is mediated by $A$. Hence $A$ actually plays the role of domain walls between regions of two different vertex types. With this in mind, the $A$ resolvent should actually correspond to placing a domain wall on the boundary. In the scaling limit we would expect this to scale to a fixed spin boundary state. 

Since the potentials are identical we can rewrite this model as \eqref{identpot} and use the above equations for the case $Q=2$ to check that this is indeed the case for the cubic potential $g V(x) = x^3/3 - x^2- 3x$.
First using \eqref{bbelowcut2} we obtain the asymptotic behaviour of $b$ as $a \rightarrow \infty$ on each sheet below the physical cut,
\bea
\label{b0inf}
b_0(a) &=& \sqrt{g a} + 1 + 2 \frac{1}{\sqrt{g a}}+ \frac{1}{2 a} + O[a^{-3/2}] \\
b_1(a) &=& -\sqrt{g a} + 1 - 2 \frac{1}{\sqrt{g a}}+ \frac{1}{2 a} + O[a^{-3/2}]. \label{b1inf}
\eea
We then use the equation \eqref{bsaddleeqn} with $b^{(1)} = b^{(2)} = b$ to compute similar expressions for the asymptotic behaviour on the sheets above the physical cut. Substituting these into \eqref{bloop} we obtain,
\beq
b^4 - 2 g a b^3  - a^3 g^3 + b^2 (-10 + g + a^2 g^2) + b (-a (-10 + g) g + K_0) + S_0 + a S_1 = 0
\eeq
where $K_i$ and $S_i$ are yet to be determined constants. We contrast this with the standard spectral curve \eqref{mastereqn}, which clearly differs. We find exact agreement if we shift $a \rightarrow (a+b)/g$. This shift is explained in the next section where we relate our construction to the standard spectral curve. However, it is interesting to pause and consider the physical interpretation of this shift; if we solve for $a(b)$, clearly it only changes the result by some analytic non-universal constant. However, solving for $b(a)$ and hence the $A$ resolvent, it corresponds to an additive renormalisation of the boundary cosmological constant by the disc function. Letting the change of variable be $a = (a'+b')/g$ and $b = b'$, we see that $b'(a') = b(a'-\frac{1}{g}b(a'+\ldots))$ which one can interpret as subtracting disc amplitudes from the boundary. Only having done this do we obtain agreement with the known spectral curve \eqref{mastereqn}. Algebraically we see it corresponds to the shift of $a$ necessary to cancel the highest order $b$ term. The lesson to learn here is that the resolvent corresponding to $M$ may require an additive renormalisation of this form. 

\subsection{The $M$ Resolvent}
We are now in a position to begin writing down the saddle point equations for \eqref{genBSA}. However, first we must apply the method in Section \ref{belowcut} to determine the analytic structure of the function $m(a)$, corresponding to the $M$ resolvent of \eqref{genBSA}, below the cut. It is straight-forward to see that $m(a) \sim a$ as $a \rightarrow \infty$ and hence has only a single sheet below the physical cut. Furthermore, using the result of Appendix \ref{AppendixA} we have that $m_0(a) = a + W_A(a)$. We can now write the saddle-point equations as,
\bea
\label{xeqns}
x_{0\ast}(a) + y_0(a) = V_A'(a) - a \\
\label{yeqns}
y_{0\ast}(a) + x_0(a) = V_A'(a) - a \\
m_{0\ast}(a) + x_0(a) + y_0(a) = V_A'(a) + W_A(a),
\eea
where we have introduced $x_0$, $x_{0\ast}$, $y_0$ and $y_{0\ast}$ for \eqref{genBSA} analogous to ${b}_0^{(1)}$, ${b}_{0\ast}^{(1)}$, ${b}_0^{(2)}$ and ${b}_{0\ast}^{(2)}$ for \eqref{genQmodel}.

\begin{figure}[t]
\centering 
\parbox{15cm}{
\centering 
\includegraphics[scale=0.5]{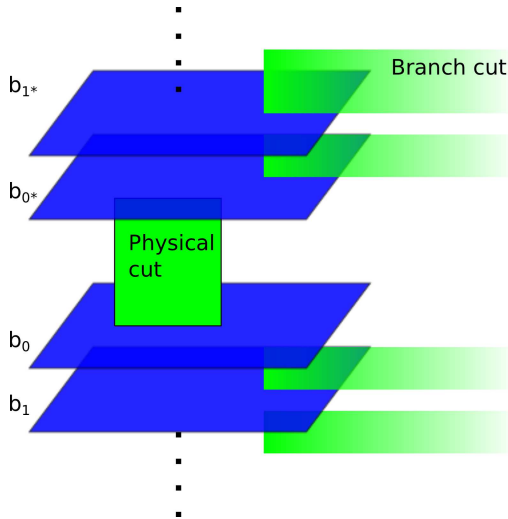}
\caption{The full analytic structure of $x(a)$. If the potentials $V_1$ and $V_2$ are of order $p$ and $q$ respectively, then the number of sheets below the physical cut is $p-1$ joined by a $p-1$ order branch cut and above the sheet we have $q$ sheets joined by an order $q-1$ branch cut. All branch cuts besides the physical cut extend to infinity.}
\label{sheets}}
\end{figure}

The equation for $m_{0\ast}(a)$ is still not particularly useful as it expresses $m_{0\ast}(a)$ in terms of $W_A(a)$ which is unknown. To rectify this, we now express $W_A(a)$ in terms of $x(a)$. First we show that $x(a)$ has the analytic structure shown in Figure \ref{sheets}. We already know that $x(a)$ has $p-1$ sheets below the physical cut, joined by an order $p-1$ branch point. Above the physical cut we see that it differs from $y(a)$ below the physical cut by an entire function. Hence the number of sheets of $x(a)$ above the physical cut equals the number of sheets of $y(a)$ below the physical cut. Furthermore, we note that the function $\sum_{i = 0}^{p-2} x_i(a)$ possess only a single branch cut, identical to the physical cut of $W_A(a)$ and therefore it differs from $W_A(a)$ only by an entire function. This is also the case for the function $\sum_{i = 0}^{q-2} y_i(a)$. Using the result of Appendix \ref{AppendixA} we find that
\beq
\label{Wa_XYsheetrelation}
\sum_{i = 0}^{p-2} x_i(a) = W_A(a) - \frac{t_{p-1}}{t_{p}} \qquad \mathrm{and} \qquad \sum_{i = 0}^{q-2} y_i(a) = W_A(a) - \frac{t'_{q-1}}{t'_{q}}.
\eeq
We therefore can write 
\beq
\label{msaddlept}
m_{0\ast}(a) = V_A'(a) + \frac{1}{2} \left( \frac{t'_{q-1}}{t'_{q}}  + \frac{t_{p-1}}{t_{p}}  \right) + \frac{1}{2} \left(\sum_{i = 1}^{p-2} x_i(a) - x_0(a)\right)  + \frac{1}{2} \left(\sum_{i = 1}^{q-2} y_i(a) - y_0(a)\right).
\eeq

Although we could obtain the spectral curve for $x(a)$ and $y(a)$ from \eqref{xeqns} and \eqref{yeqns} respectively it is easier to recall that by integrating out both the $A$ and $M$ matrix in \eqref{genBSA} we obtain the standard two matrix model to which the standard analysis may be applied. We reproduce this analysis here, since we want to express the spectral curve in terms of $x$ and $a$ rather than the usual $x$ and $y$. In particular, recalling that we may realise the $(p,p+1)$ model by setting $V_1 = V_2\equiv V$ with $V$ of order $p$, one gets
\beq
\label{MMintegrateAM}
Z = \int [dX][dY] \exp(-\frac{N}{g} \Tr[U(X)+ U(Y)  - XY]),
\eeq
where $U(X) = g V(X) - \frac{1}{2} X^2$ and $g^{-1} =  \frac{1}{\gamma}\left(1 + (\gamma-1)^{-1} \right)$. Recall that the spectral curve for the resolvent $W_X(x) \equiv \frac{1}{g}(U'(x) - \psi(x))$ is $E(x,\psi) = 0$ where $E$ is given in \eqref{mastereqn} in which $P(x,\psi)$ is a polynomial in $x$ and $\psi$ with undetermined coefficients\footnote{The undetermined coefficients being identical to those appearing in \eqref{bloop}.} apart from the coefficient of $x^{p-2} \psi^{p-2}$, where $p$ is the degree of the potential $V$. A similar equation holds for $W_Y(y)$. Given that $a(x) = V'(x) - W_X(x)$ we find the spectral curve for $x(a)$ and $y(a)$ has the form,
\beq
\label{xycurve}
E(x, ga-x) = 0 \qquad \mathrm{and} \qquad E(ga - y, y) = 0.
\eeq
This explains the shift found in the preceding section for the ``domain wall'' boundary state. From the above equation we can compute both $x(a)$ and $y(a)$ and hence, via \eqref{msaddlept}, compute $m(a)$. Finally, also setting $V_1 = V_2\equiv V$ with $V$ of order $p$ in the above saddle point equations, we may write,
\bea
\label{xeqns1mm}
x_{0 \ast}(a) + x_0(a) = V_A'(a) - a, \\
\label{meqn1}
m_{0}(a) = \left(a + \frac{t_{p-1}}{t_{p}}\right) + \sum_{i = 0}^{p-2} x_i(a),
\eea
where we have written $m(a)$ on its $0$th sheet rather than on the $0\ast$ sheet. Although \eqref{meqn1} is sufficient to compute the scaling limit of $m(a)$, it is worth noting here that \eqref{meqn1} can be used to find the spectral curve associated with the resolvent $W_M$ by the same method one uses to obtain \eqref{bloop}.

\subsection{An Example}
\label{anexample}
We give a specific example in order to illustrate the general equations introduced above. We again consider the case of $g V(x) = x^3/3 - x^2- 3x$. Using \eqref{meqn1} together with the asymptotic expansions of $x_i(a)$ as $a \rightarrow \infty$, which should be noted are identical to \eqref{b0inf} and \eqref{b1inf}, we can compute the asymptotic behaviour of $m(a)$ on each sheet above and below the cut;
\bea
m_0(a) &=& a + \frac{1}{a} + O(a^{-2}), \\
m_{0\ast}(a) &=& (1 + g) a - 2 \sqrt{ga} - 2 - \frac{4}{\sqrt{ga}} + O(a^{-3/2})\\
m_{1\ast}(a) &=& (1 + g) a + 2 \sqrt{ga} - 2 + \frac{4}{\sqrt{ga}} + O(a^{-3/2})\\
m_{2\ast}(a) &=& (1 + 2 g) a - 4 - \frac{1}{a} + O(a^{-2})
\eea
in which the sheet $m_{1\ast}(a)$ and $m_{2\ast}(a)$ are connected by a cut of finite length. Substituting into a formula for $m(a)$ equivalent to \eqref{bloop} we obtain the spectral curve for $m$ as a 4th order polynomial in $a$. This differs from the third order polynomial found for the spectral curve in \cite{Carroll:1996xi, Atkin:2010yv}, however we again find we can make an additive renormalisation of the boundary cosmological constant $m \rightarrow m + (1+g)a$ such that we may remove the 4th order term in $a$. Doing this we obtain,
\beq
\label{meqnexplicit}
a^3 - \frac{a^2}{g} \left(1 + m + \frac{1}{4}m^2 \right) + S_1 a  + S_0 + K_0 m + \frac{(2 + g) m^2}{2 g^3} + \frac{2 m^3}{g^3} + \frac{m^4}{4 g^3}=0.
\eeq
Rewriting this in terms of the resolvent and comparing to the spectral curve in \cite{Carroll:1996xi, Atkin:2010yv} we see that it contains the exact same terms, many with the exact same coefficient. We don't expect an exact agreement as we expected our construction to differ by non-universal terms. However, we should expect the basic analytic structure to match, and indeed in both cases the algebraic curve described by the above equation is a three sheeted Riemann surface with only finite branch cuts between sheets; which is quite different from the original Riemann surface for $x(a)$. The curve \eqref{meqnexplicit} is the spectral curve for the free boundary state and could in principle be used as the input for topological recursion \cite{eynardreview}. Finally, we should stress that although it was necessary to perform the additive renormalisation of $m$ to obtain agreement with \cite{Carroll:1996xi, Atkin:2010yv}, this will not actually be necessary to obtain the correct scaling limit, in the next section.

\subsection{Scaling Limit}
\label{ScaleModdp}
In order to compare the matrix model calculation to the continuum calculation we must take a scaling limit of the above results. As was stated in the preceding section it is unnecessary to know the spectral curve for $m(a)$; we can simply compute the scaling limit for $x(a)$ and use \eqref{meqn1} to find $m(a)$. 

In Daul et al.\ \cite{Daul:1993bg}, in which the method of orthogonal polynomials was used, it was shown that the two matrix model possess a scaling limit corresponding to the $(p,q)$ minimal model coupled to Liouville theory, when $U(X)$ takes a particular critical form. Unfortunately, a formulation of our current problem in terms of orthogonal polynomials is not known. Instead we have saddle-point equations from which the spectral curve may be derived.

For a given potential the scaling limit may be computed from the loop equations in the following way. First we fix the undetermined constants appearing in the spectral curve for $x(a)$. It was argued by Eynard \cite{Eynard:2002kg} that the condition necessary to fix these constants corresponds to the requirement that the spectral curve is of genus zero. Happily, this requirement means that the curve possess a rational parameterisation, i.e. there exists rational functions $\C{X}:\mathbb{C} \rightarrow \mathbb{C}$ and $\mathcal{A}:\mathbb{C} \rightarrow \mathbb{C}$ such that the points on the spectral curve are given by $(\C{X}(z), \C{A}(z))$ with $z \in \mathbb{C}$.

The form of the rational parameterisation, when it exists, for the curve $E(x,y) = 0$ is known \cite{Eynard:2002kg} and is given by,
\bea
&&\C{X}(z) = \frac{1}{z}\sum^p_{n=0} \eta_n z^n, \\
&&\C{Y}(z) = z \sum^p_{n=0} \eta_n z^{-n},
\eea
where we have already specialised the rational parameterisation found in \cite{Eynard:2002kg} to the case of identical potentials. Given the relations \eqref{xycurve} we then have $g\C{A}(z) = \C{X}(z) + \C{Y}(z)$. The dependence of the constants $\eta_n$ on the potential can be found by requiring the rational parameterisation reproduces the known asymptotic behaviour of $x(a)$ as $a \rightarrow \infty$ on each sheet. Alternatively, explicit expressions for the constants can be found in \cite{Eynard:2002kg}. As was noted in \cite{Eynard:2002kg}, the above parameterisation first appeared implicitly in \cite{Daul:1993bg} via the orthogonal polynomial approach.

In the orthogonal polynomial approach the matrix integral is written in terms of the norms of a set orthogonal polynomials $p_n(x)$. To compute the norms one needs to analyse the operators implementing multiplication and differentiation with respect to $x$. In \cite{Daul:1993bg} such operators denoted $X$, $X p_n(x) \equiv x p_n(x)$, and $P$, $P p_n(x) \equiv d_x p_n(x)$, respectively were introduced and it was shown that they could be written, in the planar limit, as Laurent series in a single parameter $z = e^\omega$.\footnote{In particular see equations (2.16) and (2.17) in \cite{Daul:1993bg}.} Identifying our parameter $z$ with the parameter $z$ in \cite{Daul:1993bg} we see $\C{X}(z)$ coincides exactly with $X(z)$, i.e in our notation $\C{X}(z)  = X(z)$.\footnote{See equation (2.10) in \cite{Daul:1993bg}.} Furthermore, it was argued in \cite{Daul:1993bg} that $P$ and the resolvent $W_X$ coincide up to non-universal terms and so $\C{Y}(z)$ will also equal $P(z)$ up to non-universal terms.

For the critical potentials appearing in \cite{Daul:1993bg} there exists a critical value for $g$, denoted $g_c$, such that the spectral curve possess a singular point $z_c$ where the first $p$ derivatives of $\C{X}$ and $\C{Y}$ vanish.\footnote{In the case of the Ising model, the critical potential corresponds to the choice, $U(x) = -3x - \frac{3}{2} x^2 + \frac{1}{3} x^3$ and $g_{c} = 10$.}  The scaling limit is then defined as the blow-up of this point \cite{eynardreview}.

Note that we do not require the first $q$ derivatives of $\C{Y}$ to vanish. This is not possible for the following reason; the minimal potentials in \cite{Daul:1993bg} ensure that there exists a single point on the curve at which enough extra zeros accumulate, or equivalently enough derivative of $\C{X}$ and $\C{Y}$ vanish at a branch cut to give a higher order scaling limit. However, in the symmetric case we have $\C{Y}(z) = \C{X}(1/z)$. Suppose $z_c \neq 1/z_c$, this would mean if $p$ derivatives of $\C{X}$ vanish at $z_c$ then $p$ derivatives of $\C{Y}$ would also vanish at $1/z_c$ contradicting the preceding property of the minimal potential. We therefore conclude that $z_c = \pm 1$ and therefore that $\C{Y}(z_c+\delta z) = \C{X}(z_c - \delta z)$ for small $\delta z$, showing that $\C{Y}$ has only $p$ derivatives which vanish at $z_c$. 

We now blow up the curve to obtain the scaling limit by setting $z = z_c + \kappa (g-g_c)^{\nu} \zeta$ and $g = g_c + \epsilon^2 \Lambda$, and letting $\epsilon \rightarrow 0$,
\bea
\label{scalingX}
\C{X}(z) &=& \C{X}(z_c) + \epsilon^{2\nu p} P(\zeta) + \epsilon^{2\nu q} Q_X(\zeta) \ldots, \\
\label{scalingY}
\C{Y}(z) &=& \C{Y}(z_c) + \epsilon^{2\nu p}  (-1)^{p} P(\zeta)  + \epsilon^{2\nu q} Q_Y(\zeta) + \ldots.
\eea
Here $\nu = 1/(p+q-1)$ and $q=p+1$ and $\kappa$ is some unimportant normalisation constant. Note that in \cite{Daul:1993bg} the scaling limit of the curve $(X(z),P(z))$ was found\footnote{See equations (5.8) and (5.9) of \cite{Daul:1993bg}.} which by comparison to the above result allows us to conclude $P(\zeta) \propto T_p(\zeta)$ and thus $T_p(-\zeta)=(-1)^p T_p(\zeta)$, which we have used to simplify \eqref{scalingY}. We may invert \eqref{scalingX} by setting $\zeta = \mu_0 + \epsilon \mu_1 + O(\epsilon^2)$ and solving perturbatively, to obtain,
\bea
P(\mu_0) &=& \epsilon^{-2\nu p} \left( \C{X}(z) - \C{X}(z_c) \right), \\
P'(\mu_1) &=& -Q_X(\mu_0),
\eea
which upon substitution into \eqref{scalingY} gives,
\beq
\C{Y}(z) - \C{Y}(z_c) = (-1)^p P(\mu_0) \epsilon^{2\nu p} + \left(Q_Y(\mu_0) - (-1)^p Q_X(\mu_0) \right) \epsilon^{2\nu q} + \ldots.
\eeq
The above equation expresses $\C{Y}(z)$ in terms of the scaling combination $\epsilon^{-2\nu p} ( \C{X}(z) - \C{X}(z_c))$ which is identified with the boundary cosmological constant. Note that the leading order term in the above equation has a scaling dimension different from the known continuum result. This is a well known issue in the matrix model approach; the leading order terms are ``non-scaling'' and the universal physics is contained in the terms at order $\epsilon^{2\nu q}$. Using the fact that the universal terms of $\C{Y}(z)$ should coincide with those of $P(z)$ appearing in \cite{Daul:1993bg} we conclude that $Q(\zeta) \equiv Q_Y(\zeta) - (-1)^p Q_X(\zeta) \propto T_q(\zeta)$. This leads naturally to choosing the parameterisation\footnote{The coefficient of $\sigma$ is chosen to agree with the literature.}  $\zeta = \cosh(\frac{\pi \sigma}{\sqrt{pq}})$. Hence for odd $p$ we have,
\bea
\label{xablowup}
\C{X}(z) &=& \C{X}(z_c) + \epsilon^{2\nu p} P(\zeta) + \ldots, \\
\C{A}(z) &=& \C{A}(z_c) + \epsilon^{2\nu q} Q(\zeta) + \ldots.
\label{xablowup2}
\eea
The case of even $p$ will be addressed later when considering a boundary magnetic field. It is now useful to introduce the scaling quantities, $\tilde{\C{X}}(\sigma) = \epsilon^{-2\nu p} (\C{X}(z) -  \C{X}(z_c))$ and $\tilde{\C{A}} (\sigma)= \epsilon^{-2\nu q} (\C{A}(z) -  \C{A}(z_c))$, so we may write the scaled curve as,
\beq
\tilde{\C{X}}(\sigma) = A_x \cosh(\frac{\pi p \sigma}{\sqrt{pq}}) \qquad \tilde{\C{A}}(\sigma) = A_a \cosh(\frac{\pi q \sigma}{\sqrt{pq}}),
\eeq
where $A_x$ and $A_a$ are constants. To find the scaled version of the curve corresponding to the $W_M$ resolvent we must find an expression for the scaled version of $x_i(a)$. Given a single value for $a$ there exist different values of $x(a)$, $x_i(a)$, corresponding to the different sheets, hence in the scaling limit we must find the values of $\sigma$ that give the same value of $a$ while changing the value of $x$. Since we have $\tilde{\C{A}}(\sigma) \propto \cosh(\frac{\pi q \sigma}{\sqrt{pq}})$, we see that $\sigma \rightarrow \sigma + 2 n i \sqrt{\frac{p}{q}}$, where $n \in \mathbb{Z}$, moves between the sheets of $x(a)$. We therefore have from \eqref{meqn1}, the scaled spectral curve,
\beq
\tilde{\C{M}}(\sigma)  = A_x  \sum_{n = 0}^{p-2} \cosh\left(\frac{\pi p}{\sqrt{pq}} \left(\sigma + 2 n i \sqrt{\frac{p}{q}}\right)\right), \qquad \tilde{\C{A}}(\sigma) = A_a \cosh(\frac{\pi q \sigma}{\sqrt{pq}}).
\eeq
Note in particular that the additive renormalisation of $m$ introduced in the previous section would not affect this result, as it would only contribute sub-leading order terms. At this point we also want to draw attention to the similarity between this relation and the Seiberg-Shih type relation appearing in the introduction. In fact we can make this similarity more manifest by making the change of variable $\sigma \rightarrow \sigma - i \sqrt{\frac{p}{q}}(p-2)$, giving,
\beq
\label{mscaled2}
\tilde{\C{M}}(\sigma)  = A_x  \sum_{n = -(p-2),2}^{p-2} \cosh\left(\frac{\pi p}{\sqrt{pq}} \left(\sigma + n i \sqrt{\frac{p}{q}}\right)\right), \qquad \tilde{\C{A}}(\sigma) = A_a (-1)^{q+1}\cosh(\frac{\pi q \sigma}{\sqrt{pq}}).
\eeq
By using the expression for the Chebyshev polynomial of the second kind, $U_{k}(\cos(x)) = \sum_{n=-k,2}^{k} \exp{(inx)}$, we see that,
\beq
\label{mscaled}
\tilde{\C{M}}(\sigma)  = A_x (-1)^{q+1} U_{p-2}(\cos(\frac{\pi}{q}))\cosh(\frac{\pi p \sigma}{\sqrt{pq}})  , \qquad \tilde{\C{A}}(\sigma) = A_a (-1)^{q+1}\cosh(\frac{\pi q \sigma}{\sqrt{pq}}),
\eeq
where we have used the fact that $p+1 =q$. To make contact with the usual form of the continuum disc amplitude it is necessary to rescale the matrix $M$ appearing in the resolvent by $(-1)^{q+1} U_{p-2}(\cos(\frac{\pi}{q}))$. This rescaling was already necessary in the work \cite{Carroll:1996xi, Atkin:2010yv}. Once rescaled, the scaled spectral curve is exactly the expression one would expect for the disc with a $(1, p-1)$ boundary state. We therefore identify this state as the free spin state of the $(p, p+1)$ model.

The expression for $\tilde{\C{M}}(\sigma)$ in \eqref{mscaled2} is very reminiscent of the Seiberg-Shih relations. However, there are some key differences. Firstly, this relation applies to resolvents rather than the non-marked disc function appearing in the original relation. Secondly, this relation arises as the scaling limit of a discrete version of the relation \eqref{meqn1}. Unfortunately we have only constructed the $(1,p-1)$ boundary state; this is a state that has been constructed before in the matrix model \cite{Ishiki:2010wb}. Our construction is nonetheless interesting, as it provides an explicit spectral curve from which this resolvent arises and therefore it is much more obvious how one might begin to apply the topological recursion \cite{eynardreview} to our construction. Indeed it appears simple to generalise the discrete expression \eqref{meqn1} so that it produces a scaling limit consistent with other boundary states,
\beq
m_{0}(a) = \left(a + \frac{t_{p-1}}{t_{p}}\right) + \sum_{i = 0}^{l-2} x_i(a),
\eeq
where $2 \leq l \leq p - 1$. 

This construction sheds light on the results in \cite{Gesser:2010fi, Atkin:2010yv} in which it was observed that the Seiberg-Shih relation no longer appears to hold for worldsheets of higher genus. It would now seem there is truth to the Seiberg-Shih relation; it gives the relation between the spectral curves of the different boundary states. However, we conjecture that to correctly compute the higher genus amplitudes one should use the appropriate objects defined on these algebraic curves in \cite{eynardreview} rather than naively applying the Seiberg-Shih relation to every amplitude.

\section{Boundary Magnetic Field}
We now consider the situation where ``spin up'' and ``spin down'' fields will be given different weights when appearing on the boundary; this should correspond to a boundary magnetic field. Particularly, we aim to generalise the boundary magnetic field studied for the Ising model in \cite{Carroll:1996xi, Atkin:2010yv} for all $(p,p+1)$ models. Consider the action,
\beq
\label{genBSAwithH}
Z = \int [dM] [dA] [dX][dY] e^{-N \Tr[V_h(X) + V_{-h}(Y) + \frac{1}{2}M^2 - (e^h X + e^{-h} Y + M)A + \frac{\gamma}{2} A^2]},
\eeq
where the potential $V_h$ has the form $V_h(x) = \frac{e^{2h}}{2g}x^2+\sum_{k=1}^\infty \frac{t_k}{k} x^k$. Integrating out $M$ and $A$ results in the two matrix model with critical potentials independent of $h$. 
Hence, ones sees that there is no $h$ dependence in the bulk. However, the $M$ resolvent will be $h$ dependent. From a similar argument as presented around \eqref{genBS} and Figure \ref{feynbound2}, one then expects that the $M$ resolvent will flow to the same boundary conditions as the $e^h X + e^{-h} Y$ resolvent. This resolvent is usually referred to as the free spin disc function with a boundary magnetic field $h$ \cite{Carroll:1996xi}. Indeed, one sees that the ``spin up'' fields (corresponding to $X$'s) adjacent to the external $M^k$ vertex in Figure \ref{feynbound2} are weighted with a factor $e^{h}$ while the ``spin down'' fields (corresponding to $Y$'s) adjacent to the external $M^k$ vertex are weighted with a factor $e^{-h}$.

We may introduce new variables $R = e^h X$ and $S = e^{-h} Y$ in order to return the action in \eqref{genBSAwithH} to the form in \eqref{genBSA}, 
\beq
\label{genBSAwithHscaled}
Z = \int [dM] [dA] [dR][dS] e^{-N \Tr[\bar{V}_h(R) + \bar{V}_{-h}(S) + \frac{1}{2} M^2 - (R + S + M)A + \frac{\gamma}{2} A^2]},
\eeq
where $\bar{V}_h(x) = V_h(e^{-h} x)$. Clearly the saddle point equation for $m(a)$ remains the same, while the equations for $r(a)$ and $s(a)$ will differ from $x(a)$ and $y(a)$. Since we will express $m(a)$ in terms of $r(a)$, we must find a rational parameterisation for the curve $r(a)$. This can be accomplished by the fact it is possible to relate $r$ to $x$ and $s$ to $y$. Upon integrating out $M$ and $A$ to obtain the standard two hermitian matrix model, 
\beq
\label{MMintegrateAM2}
Z = \int [dR][dS] \exp(-\frac{N}{g} \Tr[\bar{U}_h(R)+ \bar{U}_{-h}(S)  - RS]),
\eeq
with $\bar{U}_h(x) = g \bar{V}_h(x) - \frac{1}{2} x^2$ and $g^{-1} =  \frac{1}{\gamma}\left(1 + (\gamma-1)^{-1} \right)$, we immediately know the form of the spectral curve \cite{Eynard:2002kg},
\beq
E_h(r,s) \equiv \left(\bar{U}'_h(r) - s\right)\left(\bar{U}'_{-h}(s)- r \right) + g P_h(r,s) = 0,
\eeq
where $s(r)$ is related to the resolvent by $W_{R}(r) = \frac{1}{g}(\bar{U}'_h(r) - s(r))$. Noting that $\bar{U}'_h(x) = e^{-h} U'(e^{-h} x)$, we can write,
\beq
E_h(r,s) \propto \left(U'(e^{-h}r) - e^{h} s\right)\left(U'(e^{h} s) - e^{-h}r \right) + g P(e^{-h}r,e^{h}s) = 0,
\eeq
where we have rewritten $P_h$ by extracting the dependence on the magnetic field from the arbitrary constants. Finally given $a$ is related to $s$ in the above equations via $\bar{V}'_h(r) - a(r) = \frac{1}{g}(\bar{U}'_h(r) - s(r))$, we can write,
\beq
E\left(e^{-h} r, g e^{h} \left[\frac{1}{g}\bar{U}'_h(r) -\bar{V}'_h(r) + a(r) \right]\right) = 0,
\eeq
where $E$ is the spectral curve for the two hermitian matrix model in which both potentials are $U$. Since the curve $E$ in the above is independent of $h$, we can relate the rational parameterisation at $h \neq 0$ to the $h=0$ case,
\bea
\label{x_and_a_hneq0}
\C{R}(z;h) &=& e^{h} \C{X}(z)\\
\C{A}(z;h) &=& \bar{V}'_h(e^{h}\C{X}(z)) - e^{-h}  \left[\bar{V}'_0(\C{X}(z)) - \C{A}(z) \right ] \nn \\
&=&e^{-h}(V'_h(\C{X}(z)) - V'_0(\C{X}(z))) + e^{-h} \C{A}(z). \nn
\eea
The saddle point equations for $m$ can now be written as,
\beq
\label{SSwithH2}
\C{M}(z;h) = \left(\C{A}(z;h) + e^h \frac{t_{p-1}}{t_p}\right) + e^h \sum_{i=0}^{p-2} \C{X}(t_i;0),
\eeq
where $t_i$ corresponds to a point on the $i$th sheet of $x(a)$, i.e. given $z$ we can find $z_i$ with the same value of $a$ but differing $x$. This represents a generalisation of the Seiberg-Shih relations to the case of non-zero boundary magnetic field.

We can now also address the issue of the scaling limit in the case of even $p$. Consider \eqref{x_and_a_hneq0} when $h = i \pi/2$, we have that,
\beq
\C{A}(z;i \pi/2) = -\frac{i}{g}(g \C{A}(z) - 2 \C{X}(z) )  =  -\frac{i}{g}(\C{Y}(z) - \C{X}(z) ).
\eeq
If we substitute in the scaling form for $\C{X}$ and $\C{Y}$, i.e.\ \eqref{scalingX} and \eqref{scalingY}, we find that to leading order,
\beq
\C{A}(z;i \pi/2) = -\frac{i}{g} \epsilon^{2\nu q} Q(\zeta)
\eeq
and hence the scaling limit of $\C{M}(z;i \pi/2)$ equals the sum over $\C{X}(t_i;0)$ up to a normalisation. Hence for even $p$ and $h = i \pi/2$ we have that the scaling forms of $\C{A}(z;i \pi/2)$ and $\C{M}(z;i \pi/2)$ are equal to the scaling forms of $\C{A}(z)$ and $\C{M}(z)$ for odd $p$ with $h=0$, the analysis of Section \ref{ScaleModdp} then applies directly. We therefore conclude that for even $p$ the free spin boundary state is created by the $X-Y$ resolvent.

\subsection{Scaling Limit}
We now consider the scaling limit for general values of $h$ for odd $p$. We expect even $p$ to be identical, however it requires a slight modification of the algebra. Recall the expression \eqref{x_and_a_hneq0}. Given that we know $\C{X}(z)$ and $\C{A}(z)$ scale as \eqref{xablowup} and \eqref{xablowup2}, we can obtain the scaling form for $\C{A}(z;h)$,
\bea
\label{h_xanda_scaling}
\C{A}(z;h) &=& a_c + e^{-h} (V'_h(x_c) -  V'_0(x_c)) + e^{-h} (V''_h(x_c) -  V''_0(x_c)) A_x \cosh(\frac{\pi p \sigma}{\sqrt{pq}}) \epsilon^{2\nu p} +\nn \\
&& +\,  e^{-h}\left((V''_h(x_c) -  V''_0(x_c)) F(\sigma) +  A_a \cosh(\frac{\pi q \sigma}{\sqrt{pq}}) \right)\epsilon^{2\nu (p+1)} + \ldots  \label{scalinga},
\eea
where $F$ is an unknown function defined implicitly by $ \C{X}(z) = x_c + A_x \cosh(\frac{\pi p \sigma}{\sqrt{pq}}) \epsilon^{2\nu p} + F(\sigma) \epsilon^{2\nu (p+1)} + O(\epsilon^{2\nu (p+1)})$. The crucial point to notice about this scaling relation is that for $h\neq 0$, $a - (a_c + e^{-h} (V'_h(x_c) -  V'_0(x_c))) = O(\epsilon^{2\nu p})$, whereas for $h = 0$ this term vanishes and the leading order behaviour of $a$ is controlled by the next-to-leading order term. This property affects the sum over the sheets $x_i(a)$ in the expression for $m(a)$; the value of $x$ on its different sheets, $x_i$, are obtained from the values of $\sigma$ which map to the same $a$, but differing values of $x$. By considering the expression for $a$ in \eqref{h_xanda_scaling}, we can see given a point on the curve $\sigma$, that the points on the curve $\sigma'_n$ given by,
\beq
\sigma'_n = 
\begin{cases}
\sigma + 2 i n \sqrt{\frac{q}{p}} + \epsilon^{2\nu} K_n(\sigma)  + O(\epsilon^{4\nu}), & \text{if $h \neq 0$}\\ 
\sigma + 2 i n \sqrt{\frac{p}{q}} + O(\epsilon^{2\nu}), & \text{if $h = 0$},\\
\end{cases}
\eeq
map to the same values of $a$ to order $\epsilon^{2 \nu (p+1)}$, where $K_n$ is given by,
\beq
A_x \frac{\pi p}{\sqrt{pq}} \sinh(\frac{\pi p \sigma}{\sqrt{pq}})K_n(\sigma) = F(\sigma)-F(\sigma_n) + A_a\frac{\cosh(\frac{\pi q \sigma}{\sqrt{pq}}) - \cosh(\frac{\pi q \sigma_n}{\sqrt{pq}}) }{V''_h(x_c) -  V''_0(x_c)},
\eeq
where $\sigma_n = \sigma + 2 i n \sqrt{\frac{q}{p}}$. Substituting the $h\neq0$ case into $m$ we obtain,
\bea
\nn
\C{M}(z;h) &=& m_c + \Theta(h) A_x \cosh(\frac{\pi p \sigma}{\sqrt{pq}})\epsilon^{2 \nu p} + \Theta(h) \left[F(\sigma) + A_a\frac{\cosh(\frac{\pi q \sigma}{\sqrt{pq}})}{V''_h(x_c) -  V''_0(x_c)} \right] \epsilon^{2 \nu (p+1)}\\ 
&-& \frac{A_a \epsilon^{2 \nu (p+1)}}{V''_h(x_c) -  V''_0(x_c)} \sum^{p-2}_{n=0} \cosh(\frac{\pi q \sigma_n}{\sqrt{pq}}), \label{scalingm}
\eea
where $\Theta(h) = e^{-h}(V''_h(x_c) -  V''_0(x_c))+(p-1)e^h$. Eliminating $\sigma$ from \eqref{scalingm} and \eqref{scalinga} we find that the first non-universal term exactly reproduces the disc amplitude for the identity boundary condition. This extends the picture found in \cite{Carroll:1996xi,Atkin:2010yv} for the Ising model; for any $(p,p+1)$ model we find the boundary condition corresponding to a magnetic field has two distinct fixed points; one corresponding to free spin for $h=0$ and the other to fixed spin.

Finally, we can also comment on the dual brane found for certain values of the magnetic field in \cite{Atkin:2010yv}. It is clear from the scaling form for $m$ that for values of $h$ satisfying $\Theta(h) = 0$, $m$ must scale with a non-physical dimension. For the case $p = 3$ this is precisely the dimension associated with the dual brane, as was necessary in \cite{Atkin:2010yv} to obtain a non-trivial limit. Finally, substituting the expression for $m$ into $a$ in the case $\Theta(h) = 0$ results in precisely the dual brane amplitude. Although this qualitatively reproduces the results in \cite{Atkin:2010yv}, the precise value of $h$ for which the dual scaling limit occurs in \cite{Atkin:2010yv} was $h = i \pi/2$ which differs from that found here. A possible explanation of this can be found in our earlier discussion of the renormalisation of the boundary cosmological constant by the disc function, which was necessary in order for the spectral curve for $a(m)$ to match with that found in \cite{Atkin:2010yv}. This renormalisation took the form $m \rightarrow m  + \eta a$, which modifies \eqref{scalingm} by replacing,
\bea
\Theta(h) &\rightarrow& \Theta(h) - \eta e^{-h} (V''_h(x_c) -  V''_0(x_c)) \\
&=& (p-1) \cosh(h) + \frac{1}{g_c}(2(1-\eta) + g_c (p-1)) \sinh(h).
\eea
Requiring that we reproduce the result of \cite{Atkin:2010yv} we find $\eta = g_c (p-1)/2 + 1$. For $p = 3$ this reproduces exactly the additive renormalisation used in Section \ref{anexample}.

\section{Discussion}
In this paper we have found a number of generalisation of previously known results. Firstly, we have shown that by simply introducing an extra matrix with a Gaussian potential into the two-hermitian matrix model we are able, by computing the resolvent of this new matrix, to access a variety of new boundary states corresponding to a non-zero boundary magnetic field.  We were able to show that for the $(p,p+1)$ minimal string, the $(1,p-1)$ state corresponds to the free spin boundary and that when a non-zero boundary magnetic field is introduced all discrete states flow to the fixed spin boundary state in the continuum limit. This generalises previous work in \cite{Carroll:1996xi,Atkin:2010yv}. We found also the surprising result that for even $p$ the free spin boundary is created not by a $X+Y$-like resolvent but by a $X-Y$ resolvent.

In order to compute the resolvent corresponding to the new state we have generalised the methods of \cite{ZinnJustin:1999jg} to the case of non-identical potentials and shown how the structure of these equations leads to Seiberg-Shih like relations between the various boundary states. The key difference between the known Seiberg-Shih relations and the ones appearing in this paper is that ours are valid away from the continuum limit. Indeed, it is trivial to construct the spectral curve associated to the new boundary states. This leads to an interesting possibility for resolving the tension between the success of the Seiberg-Shih relations in the case of disc functions and their apparent failure for more complicated amplitudes; the Seiberg-Shih relations relate the algebraic curves associated with the different boundary states. If this is true it would mean that the identity brane does indeed contain all the information for constructing the other branes, however the higher genus amplitudes should be constructed by applying the topological recursion relations \cite{eynardreview} rather than applying the Seiberg-Shih relations directly. One obstruction to this program is the necessity to better understand the additive renormalisation we found necessary to bring the spectral curves found here in to a form found previously. Although this renormalisation was not necessary to obtain agreement in the scaling limit, it is unclear whether the renormalised or non-renormalised curve should be used for the topological recursion. We hope to pursue this possibility in future work.

\subsection*{Acknowledgements} The authors would like to thank J. Wheater for discussions. MA acknowledges the financial support of Universit\"{a}t Bielefeld. SZ acknowledges financial support of the STFC under grant ST/G000492/1. Furthermore, he would like to thank the Mathematical Physics Group at Universit\"{a}t Bielefeld for kind hospitality and financial support for a visit during which this work was initiated.

\appendix

\section{Appendix}
\label{AppendixA}
We prove the following useful result:
\begin{lemma}
Given a potential $V(b) = \sum_{k = 1}^{n+1} \frac{t_k}{k} b^k $ with $n \geq1$, then the asymptotic behaviour of $b(a)$ as $a \rightarrow \infty$ on a sheet below the physical cut, as determined by the equation $W_B(b(a)) + \phi a = V'(b)$, is $b(a) = (\phi/t_{n+1})^{1/n} a^{\frac{1}{n}} - \frac{t_{n}}{n t_{n+1}} + O(a^{-1/n})$. \end{lemma}
\begin{proof}
Recalling that the asymptotic behaviour we want is determined by the equation $W_B(b(a)) + \phi a = V'(b)$, we may write,
\beq
\label{basymp}
\sum_{k=0}^n t_{k+1} b^k = \phi a + O(a^{-1}),
\eeq
where we have used the known asymptotic behaviour of $W_B(b) \sim 1/b$ on its physical sheet. Making an ansatz for the asymptotic form of $b(a)$ of,
\beq
b(a) = \sum^{-\infty}_{m=1} \alpha_m a^{\frac{m}{n}}
\eeq
we have our desired result by substituting this into \eqref{basymp} and comparing coefficients of $a$.
\end{proof}

\bibliography{matrix_models}{}

\end{document}